\documentclass[12pt,sort&compress]{elsarticle}

\makeatletter
\def\ps@pprintTitle{\let\@oddhead\@empty
  \let\@evenhead\@empty
  \def\@oddfoot{\reset@font\hfil\thepage\hfil}
  \let\@evenfoot\@oddfoot
}
\makeatother

\usepackage{amsmath}
\usepackage{amsfonts}
\usepackage{amssymb}
\usepackage{amsfonts}
\usepackage{amsthm}
\usepackage{amscd}
\usepackage{mathtools}
\usepackage{commath}
\usepackage{lmodern}
\usepackage{color}
\usepackage{a4wide}
\usepackage{bm} 
\usepackage{hyperref}

\usepackage[english]{babel}
\usepackage[IL2]{fontenc}
\usepackage[utf8]{inputenc}

\newtheorem{theorem}{Theorem}[section]
\newtheorem{lemma}[theorem]{Lemma}
\newtheorem{corollary}[theorem]{Corollary}
\theoremstyle{definition}
\newtheorem{definition}[theorem]
{Definition}
\theoremstyle{remark}
\newtheorem{remark}[theorem]{\upshape\bfseries Remark}
\newtheorem{example}[theorem]{\upshape\bfseries Example}

\newcommand*{\qi}{\mathbf{i}}
\newcommand*{\qj}{\mathbf{j}}
\newcommand*{\qk}{\mathbf{k}}
\newcommand*{\Cj}[1]{{#1}^\ast}
\newcommand{\residue}[2]{\operatorname{res}_{#2}#1}
\DeclareMathOperator{\IM}{Im}
\DeclareMathOperator{\RE}{Re}

\begin{document}

\begin{frontmatter}

\title{Three Paths to Rational Curves with Rational Arc Length}

\author[1]{Hans-Peter Schröcker}
\ead{hans-peter.schroecker@uibk.ac.at}
\address[1]{University of Innsbruck, Department of Basic Sciences in Engineering Sciences, Technikerstr.~13, 6020 Innsbruck, Austria}

\author[2]{Zbyněk Šír}
\ead{zbynek.sir@karlin.mff.cuni.cz}
\address[2]{Charles University, Faculty of Mathematics and Physic, Sokolovská 83, Prague 186 75, Czech Republic}

\begin{abstract}
    We solve the so far open problem of constructing all spatial rational curves with rational arc length functions. More precisely, we present three different methods for this construction. The first method adapts a recent approach of (Kalkan et al. 2022) to rational PH curves and requires solving a modestly sized system of linear equations. The second constructs the curve by imposing zero-residue conditions, thus extending ideas of previous papers by (Farouki and Sakkalis 2019) and the authors themselves (Schröcker and Šír 2023). The third method generalizes the dual approach of (Pottmann 1995) from planar to spatial curves. The three methods share the same quaternion based representation in which not only the PH curve but also its arc length function are compactly expressed. We also present a new proof based on the quaternion polynomial factorization theory of the well known characterization of the Pythagorean quadruples.
 \end{abstract}


\end{frontmatter}

\section{Introduction}
\label{sec:introduction}

Various applications require rational parametric curves and some of them profit from rationality of the curve's unit tangent field. This observation leads to the introduction of planar PH curves  (``Pythagorean hodograph curves'') in \cite{farouki90c} and spatial PH curves in \cite{farouki94a}. A vast majority of research papers focuses on \emph{polynomial} PH curves which have the advantage of a direct construction via integration of the hodograph \cite{farouki08}. All polynomial PH curves have also polynomial arc length functions. This fact was exploited in various important constructions, see e.g. \cite{farouki16,knez22,FAROUKI2021125653,knez24}. 

Let us now discuss the methods for the construction of rational PH curves. The problem is very different since integration of a rational hodograph does not necessarily generate a rational curve. This problem was bypassed in \cite{Pottmann95} where all planar rational PH curves were constructed via a dual geometrical approach. A generalization of this dual approach to the spatial rational PH curves appeared in \cite{FaroukiSir} and was later improved and exploited in \cite{FaroukiSir2,krajnc1,krajnc2,krajnc3}. A different approach of a more algebraic flavor was used in \cite{LEE2014689} to construct a particular kind of planar rational PH curves. All spatial rational PH curves (containing the planar ones) were algebraically constructed in \cite{kalkan22,schroecker23} via solving a system of linear equations and in \cite{SchroeckerSir:_optimal_interpolation} by imposing zero residue conditions on hodograph.

Another research topic needs to be discussed in connection with rational curves. Whereas all polynomial PH curves have polynomial arc length functions, only a \emph{proper subset} of the rational PH curves admits rational arc lengths. Determining this subset seems to be a rather difficult problem. From the algebraic point of view several initial results were of negative nature showing e.g.\ that rational curves other then straight lines cannot have a rational arc length parameterization \cite{FAROUKI1991151,FAROUKI2007238,SAKKALIS2009494}. Concerning the arc length function of the planar curves constructed in \cite{LEE2014689}, the analysis by \cite{FAROUKI20151} revealed that only in very particular cases these curves have a rational arc length.

Only two papers provide methods for construction of rational PH curves with rational arc-length function. In the planar case, \cite{Pottmann95} contains constructions for curves with rational arc length as evolutes of planar rational PH curves and as projections of curves of constant slope (curves whose tangents form a constant angle with a certain fixed direction, usually thought of as the vertical direction). The resulting formulas are however quite complicated due to the intrinsic dual nature of the construction. In \cite{FaroukiSakkalis2019} a special form of the rational hodograph is used to construct a subset of planar rational PH curves with rational arc length function. An extension of this method to the spatial PH curves is also hinted at and used to construct one example. To the best of our knowledge, this example is the only positive result about spatial rational curves with rational arc length so far. We devote Subsection~\ref{FarSak} of the present paper to a comparison of our results to \cite{FaroukiSakkalis2019}.

In this paper we present three different methods for constructing rational curves with rational arc length. All of them are \emph{universal} in the sense that they produce \emph{all} rational curves with that property. The first method adapts a recent approach of \cite{kalkan22} for computing rational PH curves and requires solving a modestly sized system of linear equations. The second constructs the curves by imposing linear zero-residue conditions on the hodograph, thus extending ideas of previous papers \cite{FaroukiSakkalis2019,SchroeckerSir:_optimal_interpolation}. The third method generalizes the dual approach of \cite{Pottmann95} from planar to spatial curves. We compare these three methods with previous approaches, discuss their advantages and disadvantages and also comment on aspects of genericity.

The remainder of this paper is organized as follows. In Section~\ref{sec:preliminaries} we review some elementary facts about polynomial and rational PH curves. We also give a new proof of the characterizations of PH quadruples \cite{Dietz,choi02,farouki02} that is based on the factorization theory of quaternion polynomials \cite{niven41,gordon65,hegedus13}. Two different algebraic methods for construction of curves with rational arc length are presented in Section~\ref{sec:algebraic-approach}. A short comparison of these methods with \cite{FaroukiSakkalis2019} is given in Subsection~\ref{FarSak}. The same quaternion expression $\mathcal A(t) (1 + \qi) \Cj{\mathcal A}(t)$ naturally occurs in both of our algebraic methods and motivates the analysis and geometric construction of the curves of constant slope in Section~\ref{sec:geometric-approach}. Finally, we discuss and compare the different approaches in Section~\ref{sec:comparison} and conclude the paper in Section~\ref{sec:conclusion}.

\section{Preliminaries}
\label{sec:preliminaries}

The rational parametric curve $\mathbf{r}(t)$ has the arc length function $s(t) = \int \sqrt{\dot{\mathbf{r}}(t) \cdot \dot{\mathbf{r}}(t)} \dif t$. It is called \emph{Pythagorean Hodograph} (or PH, for short) if the speed function $\sqrt{\dot{\mathbf r}(t) \cdot \dot{\mathbf r}(t)}$ is piecewise rational. This leads to a first simple but fundamental observation:

\begin{lemma}
  \label{lem:PH}
  A rational parametric curve $\mathbf{r}(t)$ with piecewise rational arc length function is necessarily a Pythagorean Hodograph curve.
\end{lemma}

A polynomial curve $\mathbf r(t)$ is PH if there exists a quadruple of polynomials $x(t)$, $y(t)$, $z(t)$, $\sigma(t)$ satisfying the Pythagorean condition
\begin{equation}
  \label{eq:1}
  \dot{x}(t)^2 +\dot{y}(t)^2 +\dot{z}(t)^2=\sigma(t)^2
\end{equation}
such that $\mathbf r(t) = [x(t), y(t), z(t)]$ \cite{farouki08}. A sufficient and necessary condition for satisfying \eqref{eq:1} was given in \cite{Dietz} (see also \cite{farouki02}) in terms of arbitrary polynomials $u(t)$, $v(t)$, $p(t)$, $q(t)$, $w(t)$ such that
\begin{equation}
  \label{eq:2}
  \begin{aligned}
    \dot{x}(t) &= w(t)[u^2(t)+v^2(t)-p^2(t)-q^2(t)],\\
    \dot{y}(t) &= w(t)[2u(t)q(t)+2v(t)p(t)],\\
    \dot{z}(t) &= w(t)[2v(t)q(t)-2u(t)p(t)],\\
    \sigma(t) &= w(t)[u^2(t)+v^2(t)+p^2(t)+q^2(t)].
  \end{aligned}
\end{equation}
Note that the polynomial $w(t)$ is a real factor of the derivative vector $\dot{\mathbf r}(t) = [\dot{x}(t), \dot{y}(t), \dot{z}(t)]$ as well as of the speed function $\sigma(t)$ and is usually omitted in applications. It must, however, be considered in order fully solve \eqref{eq:1} because there is no other way to produce real roots of~$\sigma(t)$ of odd multiplicity.

The representation result \eqref{eq:2} is relevant also for rational PH curves. By expressing the three derivative components using a common denominator it is obvious that any rational PH curve $\mathbf r(t)$ satisfies
\begin{equation}
  \label{eq:3}
  \dot{\mathbf r}(t)=\lambda(t)[\dot{x}(t), \dot{y}(t), \dot{z}(t)],
\end{equation}
for some polynomials $x(t)$, $y(t)$, $z(t)$ satisfying the Pythagorean condition \eqref{eq:1} and a rational function $\lambda(t)$. However, not all rational functions $\lambda(t)$ give rise to rational PH curve via integration of \eqref{eq:3}, because also certain logarithm and arctangent functions have rational derivatives. On the other hand if $\lambda(t)=w(t)$ is polynomial we get precisely \eqref{eq:2}, thus producing all polynomial PH curves.

It is convenient and customary in the context of PH curves (cf.~\cite{choi02b}) to represent spatial parametric curves and their derivatives as functions of a real variable with values in the algebra $\mathbb{H}$ of quaternions and in particular in its imaginary part (or vector part) $\IM \mathbb{H}$ which is identified with $\mathbb R^3$. We denote by $\mathbb{H}[t]$ the algebra of polynomials in $t$ and with coefficients in $\mathbb{H}$. Since $t$ is a real parameter, multiplication of quaternions is defined by the convention that $t$ commutes with all coefficients. In this way the polynomial hodograph $[\dot{x}(t), \dot{y}(t), \dot{z}(t)]$ would be identified with the quaternion polynomial
\begin{equation*}
  \mathbf F(t)=\dot{x}(t)\qi +\dot{y}(t)\qj + \dot{z}(t)\qk \in \IM \mathbb{H}[t].
\end{equation*}
Denote by $\Cj{\mathbf q}$ the conjugate quaternion and by $N\colon \mathbb{H} \to \mathbb{R}_{\ge 0}$, $\mathbf q \mapsto \mathbf q \Cj{\mathbf q}$ the quaternion norm function. Extending it to a map from $\mathbb{H}[t]$ into the set of positive univariate polynomials, Equation~\eqref{eq:1} then becomes $N(\mathbf F(t))=\sigma^2(t)$.

Let us recall the following definition (first presented in \cite[Definition 3.4]{kalkan22}) which technically simplifies the construction of PH curves.

\begin{definition}
  \label{def:1}
  We say that a quaternion polynomial $\mathcal A(t)$ is \emph{reduced with respect to $h \in \mathbb{H}$}, or \emph{$h$-reduced} for short, if it is free of non-constant real factors and of polynomial right factors with coefficients in the sub-algebra of $\mathbb{H}$ which is generated by $1$ and~$h$.
\end{definition}

The basic idea behind Definition~\ref{def:1} is to avoid spurious real factors in quaternionic polynomials of the shape $\mathcal A(t) h \Cj{\mathcal A}(t)$. Indeed, if $\mathcal A(t) = \mathcal B(t) \mathcal H(t)$ with $\mathcal H(t)$ in the sub-algebra generated by $1$ and $h$, then $h$ and $\mathcal H(t)$ commute so that
\begin{equation*}
  \mathcal A(t) h \Cj{\mathcal A}(t) =
  \mathcal B(t) h \Cj{\mathcal B}(t) \mathcal{H}(t) \Cj{\mathcal H}(t)
\end{equation*}
and $\mathcal H(t) \Cj{\mathcal H}(t)$ is a real polynomial. As a consequence of \cite[Proposition~2.1]{cheng16}, also the inverse implication is true: $\mathcal{A}(t)h\Cj{\mathcal{A}}(t)$ has a non-trivial real polynomial factor if and only if $\mathcal{A}(t)$ is not $h$-reduced. This can be checked algorithmically by showing that the coefficients of $\mathcal{A}(t)h\Cj{\mathcal{A}}(t)$ have a non-trivial~$\gcd$.

We continue by expressing the characterization of Pythagorean quadruples \eqref{eq:2} using quaternion polynomials \cite{choi02b} and present a new constructive proof based on the factorization theory of quaternion polynomials \cite{niven41,gordon65,hegedus13}.

\begin{lemma}
  \label{lem:factorization}
  Let $\mathbf F(t)\in \IM \mathbb{H}[t]$ be without real polynomial factors such that $N(\mathbf F(t))=\sigma^2(t)$. Then there exists $\mathcal A(t)\in \mathbb{H}[t]$, reduced with respect to $\qi$, so that
  \begin{equation*}
    \mathbf F(t)=\mathcal A(t) \qi \Cj{\mathcal  A}(t).
  \end{equation*}
\end{lemma}

\begin{proof}
  Since $\mathbf F(t)$ has no real factors, the real polynomial $\sigma$ is the product of irreducible quadratic real polynomials and the degree $n=\deg \mathbf F$ is even. We pick one such factor $M_1(t)$ and assume without loss of generality that it is monic. There exists a quaternion $h_1 \in \mathbb{H}$ such that $t-h_1$ is a left factor of $\mathbf F(t)$, that is, $\mathbf F(t) = (t - h_1) \tilde{\mathbf F}(t)$ for some polynomial $\tilde{\mathbf F}(t) \in \mathbb{H}[t]$, and $M_1(t) = N(t-h_1)$ \cite[Lemma~3]{hegedus13}. Since $N(\mathbf F(t))$ is a square, $M_1(t)$ is also a factor of $N(\tilde{\mathbf F}(t))$ and a symmetric argument yields existence of a linear right factor $t - \overline{h}_1$ of $\tilde{\mathbf F}(t)$, that is $\mathbf F(t) = (t - h_1)\tilde{\mathbf F}(t) = (t - h_1) \overline{\mathbf F}(t) (t - \overline{h}_1)$ for some polynomial $\overline{\mathbf F}(t) \in \mathbb{H}[t]$. Now
  \begin{equation*}
    (t - \Cj{\overline{h}}_1) \Cj{\overline{\mathbf F}}(t) (t - \Cj{h}_1) =
    \Cj{\mathbf F}(t) =
    -\mathbf F(t) =
    -(t - h_1) \overline{\mathbf F}(t) (t - \overline{h}_1)
  \end{equation*}
  so that uniqueness of the linear left and right factors yields $h_1 = \Cj{\overline{h}}_1$. Moreover, we have
  \begin{equation*}
    0
    = 2\RE \mathbf F(t)
    = \mathbf F(t) - \Cj{\mathbf F}(t)
    = (t - h_1)(\overline{\mathbf F}(t) - \Cj{\overline{\mathbf F}}(t))(t - \Cj{h}_1)
  \end{equation*}
  whence $\overline{\mathbf F}(t) - \Cj{\overline{\mathbf F}}(t) = 0$ and $\overline{\mathbf F}(t) \in \IM\mathbb{H}[t]$. Proceeding inductively, we find $\mathbf F(t) = \mathcal A_0(t) f_n \Cj{\mathcal A}_0(t)$ where $\mathcal A_0(t) = (t - h_1) \cdots (t - h_{n/2})$, $\deg \mathbf{F} = n$, and $f_n$ is the leading coefficient of $\mathbf F(t)$. Since there exists $x \in \mathbb{H}$ such that $f_n = x\qi\Cj{x}$ \cite[Lemma~1]{SirC1}, the claim follows with $\mathcal A(t) = \mathcal{A}_0(t) x$. The thus constructed polynomial $\mathcal A(t)$ is $\qi$-reduced as otherwise $\mathbf F(t)$ had a real polynomial factor.
\end{proof}

Summarizing results so far, we state:

\begin{corollary}
  \label{cor:1}
  A spatial rational curve $\mathbf r(t)$ is a PH curve with piecewise rational arc length if and only if there exist a polynomial $\mathcal A(t)$, reduced with respect to $\qi$, and a rational function $\lambda(t)$ so that the integrals
  \begin{align}
    \label{eq:4}
    \mathbf r(t)&=\int \lambda(t) \mathcal A(t) \qi \Cj{\mathcal  A}(t) \dif t,\\
    \label{eq:5}
    s(t)&=\int \lambda(t) \mathcal A(t) \Cj{\mathcal  A}(t) \dif t
  \end{align}
  are both rational.
\end{corollary}

\begin{remark}
  Strictly speaking the term $|\lambda(t)|$ should occur in integral \eqref{eq:5} because of possible change of sign of $\lambda(t)$ at its real roots. This is also the reason why $\sqrt{\dot{\mathbf r}(t) \cdot \dot{\mathbf r}(t)}$ is required to be \emph{piecewise} rational in the definition of PH curves. Because the roots of $\lambda(t)$ correspond to the points without properly defined tangent vector we will from now on tacitly assume that the parameter domain is restricted to an interval with $\lambda(t)>0$ and use the word ``rational'' instead of ``piecewise rational''. Note also that the assumption of $\mathcal A(t)$ being reduced with respect to $\qi$ does not induce any loss of generality. It follows from the proof of Lemma~\ref{lem:factorization} that $w(t)$ is needed in \eqref{eq:2} only in order to produce real roots of $\sigma(t)$. In the case of rational curves these roots can be produced by $\lambda(t)$.
\end{remark}

Ensuring the rationality of the integral \eqref{eq:4} or bypassing integration by some other considerations is the main difficulty in the construction of rational PH curves. In the past, is was solved by various methods, e.g.\ in \cite{FaroukiSir,FaroukiSir2,kalkan22,schroecker23,SchroeckerSir:_optimal_interpolation}. The purpose of the present paper is to extend these methods so that both integrals \eqref{eq:4} and \eqref{eq:5} are rational; this has been addressed already in \cite{FaroukiSakkalis2019}.

\section{Algebraic Approaches to Rational Arc Length Curves}
\label{sec:algebraic-approach}

In this section we extend two recent algebraic approaches for the computation of rational PH curves to curves with a rational arc length function.

\subsection{First Way -- Solving a System of Linear Equations}
\label{sec:second-path}

In \cite{kalkan22,schroecker23}, the authors computed rational PH curves by an approach which can be seen as the most direct way possible. The parametric curve is represented explicitly as
\begin{equation*}
  \mathbf r(t)=\frac{\mathbf B(t)}{\alpha(t)},
\end{equation*}
where the numerator polynomial has three components ($\mathbf B(t) \in \IM\mathbb{H}[t]$) and $\alpha(t)$ is a real polynomial. In view of Corollary~\ref{cor:1}, this curve is PH if and only if there exists an $\qi$-reduced polynomial ${\mathcal A}(t) \in \mathbb{H}[t]$ so that $\dot{\mathbf r}(t)$ and $\mathbf F(t) = \mathcal A(t) \qi \Cj{\mathcal A}(t) \in \IM\mathbb{H}[t]$ are linearly dependent for each $t$. This dependence was equivalently characterized in \cite{kalkan22,schroecker23} by existence of a polynomial $\mu(t) \in \mathbb{R}[t]$ such that
\begin{equation}
  \label{eq:6}
  \alpha(t) \dot{\mathbf B}(t) - \dot\alpha(t) \mathbf B(t) = \mu(t) \mathbf F(t) = \mu(t) \mathcal A(t) \qi \Cj{\mathcal A}(t).
\end{equation}
Note, that there is no need to consider $\mu(t)$ to be rational (with non trivial denominator) thanks to $\mathcal A(t)$ being reduced with respect to $\qi$ which implies that $\mathbf F(t)$ is free of real factors. As we will see below, the polynomial $\lambda(t)$ of \eqref{eq:4} satisfies $\lambda(t)=\frac{\mu(t)}{\alpha^2(t)}$.

If, for any given polynomials $\mathcal A(t)$ and $\alpha(t)$, we fix the respective degrees of $\mathbf B(t)$ and $\mu(t)$, Equation~\eqref{eq:6} becomes a finite system of linear equations. Solving this system for the unknown coefficients of $\mathbf{B}(t)$ and $\mu(t)$, we obtain the rational PH curve $\mathbf r(t) = \frac{\mathbf B(t)}{\alpha(t)}$.\footnote{The solution for $\mu(t)$ is not required for the problem at hand and may be discarded. It can, however, be useful for some applications, c.f. \cite{SchroeckerSir:_optimal_interpolation} or our discussion in Section~\ref{sec:comparison}.} Note that it is not guaranteed that $\mathbf B(t)$ and $\alpha(t)$ are free of common real factors. In fact, there are strong conditions on $\alpha(t)$ and $\mathbf F(t)$ in order to allow for non-polynomial solutions \cite[Theorem~4.6]{kalkan22}.

Let us investigate the condition for $\mathbf r(t)$ to have a piecewise rational arc length. We have
\begin{equation}
  \label{eq:7}
  \dot{\mathbf r}(t) = \frac{\alpha(t) \dot{\mathbf B}(t) - \dot\alpha(t) \mathbf B(t)}{\alpha^2(t)} = \frac{\mu(t)\mathcal A(t) \qi \Cj{\mathcal A}(t)}{\alpha^2(t)}
\end{equation}
whence the arc length integral becomes
\begin{equation*}
  s(t) = \int \Vert \dot{\mathbf r}(t) \Vert \dif t
       = \int \vert \mu(t) \vert \alpha^{-2}(t) \mathcal A(t) \Cj{\mathcal A}(t) \dif t.
\end{equation*}
On intervals where $\mu(t)$ is strictly positive (the argument for polynomials that are negative everywhere is similar), this simplifies to
\begin{equation*}
  s(t) = \int \mu(t) \alpha^{-2}(t) \mathcal A(t) \Cj{\mathcal A}(t) \dif t.
\end{equation*}
For $s(t)$ to be rational, the integrand should be the derivative of a rational function, that is,
\begin{equation*}
  \dot{s}(t) = \frac{\beta(t)}{\alpha^2(t)}
  \quad\text{where}\quad
  \beta(t) = \mu(t) \mathcal A(t) \Cj{\mathcal A}(t).
\end{equation*}

Comparing with \eqref{eq:7}, we see that we need to augment Equation~\eqref{eq:6} with one additional condition in order to compute a suitable rational function $s(t)$. We write
\begin{equation}
  \label{eq:8}
  \alpha(t) \dot{\mathcal B}(t) - \dot\alpha(t)\mathcal B(t) = \mu(t)\mathcal A(t) (1 + \qi) \Cj{\mathcal A}(t).
\end{equation}
In contrast to \eqref{eq:6}, the right-hand side is no longer vectorial but an element of $\mathbb H[t]$ and the same is true for any of its solutions $\mathcal B(t)$. The rational curve $\mathcal R(t) = \frac{\mathcal B(t)}{\alpha(t)}$ lives in $\mathbb{H} \cong \mathbb R^4$, its vector part $\mathbf r(t) \coloneqq \IM\mathcal R(t)$ is a rational PH curve with rational arc length $s(t) = \RE\mathcal R(t)$. A summary of this procedure plus a condition for the resulting curve to be non-polynomial is in

\begin{theorem}
  \label{th:linear-system}
  Given polynomials $\alpha \in \mathbb R[t]$ and $\mathcal A(t) \in \mathbb{H}[t]$, solutions of \eqref{eq:8} for $\mathcal B(t) \in \mathbb{H}[t]$ and $\mu(t) \in \mathbb{R}[t]$ yield rational curves $\mathcal R(t) = \frac{\mathcal B(t)}{\alpha(t)}$ in $\mathbb{H}$ whose vector parts $\mathbf r(t) = \IM \mathcal R(t)$ have a rational arc length function. Assuming without loss of generality that $\mathcal{A}(t)$ is $\qi$-reduced, non-polynomial solutions exist for sufficiently high degree of $\mathcal{B}(t)$ if and only if there is a zero $z$ of $\alpha$ with multiplicity $n$ such that the set $\{\mathcal F_0, \mathcal F_1, \ldots, \mathcal F_n \}$ with coefficient $\mathcal F_i$ defined by $$\mathcal F(t)=\mathcal A(t) (1 + \qi) \Cj{\mathcal A}(t) = \sum_i (t - z)^i \mathcal F_i$$ is linearly dependent over $\mathbb{C}$. This condition is always satisfied if $n \ge 4$.
\end{theorem}

\begin{remark}
  Theorem~\ref{th:linear-system} is an extension of \cite[Theorem~4.6]{kalkan22} with the expression $\mathcal F(t) = \mathcal{A}(t)\qi\Cj{\mathcal A}(t)$ replaced by $\mathcal F(t) = \mathcal{A}(t)(1+\qi)\Cj{\mathcal A}(t)$. Consequently, the coefficients of the Taylor expansion of $\mathcal F(t)$ at $t = z$ live in $\IM \mathbb{H} = \mathbb R^3$ and the condition for existence of non-polynomial solutions in \cite[Theorem~4.6]{kalkan22} generically requires only a zero of multiplicity $n \ge 3$.
\end{remark}

A straightforward proof of Theorem~\ref{th:linear-system} along the lines of the proof of \cite[Theorem~4.6]{kalkan22} would be possible but, in the context of this article, is unnecessarily long. We prefer to view Theorem~\ref{th:linear-system} as a consequence of Theorem~\ref{th:zero-residue} below and provide a short proof in Section~\ref{sec:zero-residue}.

\begin{example}
  \label{ex:linear-system}
  Let us consider the polynomials
  \begin{equation}
    \label{eq:9}
    \mathcal{A}(t) = t^2 + \qk t + \qi + \qj
    \quad\text{and}\quad
    \alpha(t) = t^4.
  \end{equation}
  Non-constant solutions require $\deg \mathbf{B}(t) \ge 8$. Using this minimal degree, the solution is unique up to scaling and translation. We compute it not by solving \eqref{eq:8} directly but by equating the coefficients of powers of $t$ in the $2 \times 2$ minors of the matrix with columns $\alpha(t)\dot{\mathcal B}(t) - \dot{\alpha}(t)\mathcal B(t)$ and $\mathcal A(t) (1 + \qi) \Cj{\mathcal A}(t)$. (This is a more straightforward computation but we found it less suitable for theoretically analyzing the resulting system of linear equations.) One solution is
  \begin{multline}
    \label{eq:10}
    \mathbf r(t) = \frac{1}{t^4}
    \bigl(
    75 \qi t^8
    - 100 (\qi - 2\qj)t^7
    - 300 (\qi + \qj + \qk)t^6
    - 600 (\qj - 2 \qk)t^5\\
    - 32 (3 \qi + 5 \qj - 15 \qk)t^4
    - 300 ( \qi - 4 \qj)t^3
    - 150 (\qi - 2 \qj)t^2
    + 200 (\qj + \qk)t
    + 150 \qj
    \bigr).
  \end{multline}
  All other solutions are obtained from \eqref{eq:10} by scaling and translation.
\end{example}

\subsection{Second Way -- Zero Residue Conditions}
\label{sec:zero-residue}

The paper \cite{SchroeckerSir:_optimal_interpolation} presents a method to compute rational PH curves by directly studying the rationality of the integral \eqref{eq:4} which can be insured by imposing ``zero residue'' conditions on the integrand. These are formulated as linear equations in terms of the partial fraction coefficients of the rational function $\lambda(t)$ of Equations~\eqref{eq:4} and \eqref{eq:5}. More precisely, performing partial fraction decomposition, we can write
\begin{equation}
  \label{eq:11}
  \lambda(t) = p(t) + \sum_{i=1}^m \sum_{j = -k_i}^{-1} \lambda_{i,j}(t - \beta_i)^j
\end{equation}
where $\prod_{i=1}^m(t-\beta_i)^{k_i}$ is the denominator of $\lambda(t)$ in reduced form, $\beta_1$, $\beta_2$, \ldots, $\beta_m \in \mathbb C$ are its pairwise different zeros and $k_1$, $k_2$, \ldots, $k_m$ their respective multiplicities, $p(t)$ is a polynomial and $\lambda_{i,j}$ are uniquely determined complex coefficients. The $m$ linear zero residue conditions that are equivalent to the rationality of the integral \cite[Equation~(7)]{SchroeckerSir:_optimal_interpolation} read
\begin{equation}
  \label{eq:12}
  \lambda_{i,-1}\mathbf f_{i,0} + \lambda_{i,-2}\mathbf f_{i,1} + \lambda_{i,-3}\mathbf f_{i,2} + \cdots + \lambda_{i,-n-1}\mathbf f_{i,n} = 0,
  \quad i \in \{ 1, 2, \ldots, m\}
\end{equation}
where $n=\deg \mathbf F(t)$ and $\sum_{j=0}^n \mathbf f_{i,j}(t - \beta_i)^j$ is the Taylor polynomial of $\mathbf F(t)$ at~$\beta_i$. Note that even if \eqref{eq:12} is a system of linear equations over the complex numbers, the solution polynomial $\lambda(t)$ and the resulting rational curve will be real if the non-real zeros of $\lambda$ come in conjugate complex pairs with equal respective multiplicities.

In order to obtain not only a rational curve but also a rational arc length function we must impose also the rationality of the real integral \eqref{eq:5}. For this purpose the residue vector condition \eqref{eq:12} is to be augmented with the residue of the integrand in \eqref{eq:5}. Quite remarkably this can be done using the same expression as in Section~\ref{sec:second-path}. Defining $\mathcal F(t) \coloneqq \mathcal{A}(t)(1+\qi)\Cj{\mathcal A}(t)$ and considering the Taylor expansion
\begin{equation*}
  \mathcal F(t) = \sum_{j=0}^n \mathcal F_{i,j}(t - \beta_i)^j
\end{equation*}
for $i \in \{1,2,\ldots,m\}$ we obtain the updated zero residue conditions
\begin{equation}
  \label{eq:13}
  \lambda_{i,-1}\mathcal F_{i,0} + \lambda_{i,-2}\mathcal F_{i,1} + \lambda_{i,-3}\mathcal F_{i,2} + \cdots + \lambda_{i,-n-1}\mathcal F_{i,n} = 0,
  \quad i \in \{ 1, 2, \ldots, m\}.
\end{equation}
Note that the vector part of $\mathcal F_{i,j}$ is precisely $\mathbf f_{i,j}$ and thus there is for each $i \in \{1,2,\ldots,m\}$ precisely one additional equation in \eqref{eq:13} comparing to \eqref{eq:12}. This leads to

\begin{theorem}
  \label{th:zero-residue}
  Equation~\eqref{eq:4} yields a rational curve with rational arc length if and only if the rational function $\lambda(t)$ satisfies the zero residue conditions \eqref{eq:13}. All rational curves with rational arc length can be obtained in that way.
\end{theorem}

Now we are also in a position to prove Theorem~\ref{th:linear-system}:

\begin{proof}[Proof of Theorem~\ref{th:linear-system}]
  Each of the $m$ vector equations in \eqref{eq:13} imposes four linear constraints on the coefficients of $\lambda(t)$ which can be satisfied in a non-trivial way if and only if for at least one value $i \in \{1,2,\ldots,m\}$ the set $\{\mathcal F_{i,0}, \mathcal{F}_{i,1}, \ldots, \mathcal{F}_{i,k_i}\}$ is linearly dependent. This is guaranteed if $k_i \ge 4$ since five quaternions are always linearly dependent over~$\mathbb{C}$.
\end{proof}

\begin{remark}
  \label{rem:dimension}
  The formulations and proofs of Theorems~\ref{th:linear-system} and \ref{th:zero-residue} can be adapted to the case of rational curves with rational arc length in arbitrary dimension with a prescribed (reduced) polynomial tangent field. It turns out that non-polynomial solutions can be obtained only under some specific conditions concerning the multiplicity of the denominator roots. For generic spatial input data this multiplicity must be equal or greater to $4$. It can be less in non-generic cases exhibiting linear dependencies of the derivatives of $\mathbf F$ at the root value.
\end{remark}

\begin{example}
  \label{ex:zero-residue}
  The purpose of this extensive example is manifold. We wish to recover the solution of Example~\ref{ex:linear-system} via the zero residue approach. But we also want to further elucidate the relation between rational PH solutions and solutions that have a rational arc length (and are necessarily PH). In order to do so, we compute bases of the corresponding solution spaces. By $\residue{f(t)}{t_0}$ we denote the residue of the real or vector valued function $f(t)$ at $t = t_0$.

  We define $\mathcal A(t)$ as in Example~\ref{ex:linear-system} and set
  \begin{equation}
    \label{eq:14}
    \lambda(t) = \sum_{\ell=-\infty}^\infty \lambda_\ell t^\ell
  \end{equation}
  for the yet to be determined rational function $\lambda(t)$. Only finitely many of the coefficients $\lambda_\ell$ are different from zero. This choice of $\lambda(t)$ allows to construct rational PH curve where $t = 0$ is the only zero of the denominator polynomial. Thus, there is also just one quaternionic zero residue condition. In case of a rational PH curve, it reads $\residue{\lambda(t)\mathcal{A}(t)\qi\Cj{\mathcal{A}}(t)}{0} = 0$ or, more explicitly,
  \begin{equation}
    \label{eq:15}
    \lambda_{-5} \qi + 2 \lambda_{-4} \qj + \lambda_{-3} (-\qi - 2\qk) + 2 \lambda_{-2} \qk + 2 \lambda_{-1} \qj = 0.
  \end{equation}
  Note that only finitely many unknowns enter this condition. Their total number equals $2\deg\mathcal{A}(t)+1$ in general and five in our example. Moreover, all coefficients are \emph{vectorial quaternions.}

  In case of a rational arc length curve, the zero residue condition reads $\residue{\lambda(t)\mathcal{A}(t)(1+\qi)\Cj{\mathcal{A}}(t)}{0} = 0$ or, more explicitly,
  \begin{equation}
    \label{eq:16}
    \lambda_{-5} (1 + \qi) + 2 \lambda_{-4} \qj + \lambda_{-3} (-\qi - 2\qk + 1) + 2 \lambda_{-2} \qk + \lambda_{-1} (2\qj + 2) = 0.
  \end{equation}
  Here, the coefficients are general quaternions but their vector parts agree with the coefficients in \eqref{eq:15}. The quaternionic Equation~\eqref{eq:15} gives rise to three scalar conditions for the unknowns $\lambda_\ell$, namely,
  \begin{equation}
    \label{eq:17}
    \lambda_{-5} - \lambda_{-3} = 2 \lambda_{-4} + 2 \lambda_{-1} = -2 \lambda_{-3} + 2 \lambda_{-2} = 0.
  \end{equation}
  In case of rational arc length curves, this system is to be augmented with yet another linear equation arising from \eqref{eq:16},
  \begin{equation}
    \label{eq:18}
    \lambda_{-5} + \lambda_{-3} + 2 \lambda_{-1} = 0.
  \end{equation}

  Now we are going to describe basis vectors $\mathbf{p}_\ell(t)$ of the solution space of rational PH curves and basis vectors $\mathbf{r}_\ell(t)$ of the solution space of rational arc length curves.

  In order to obtain solutions $\mathbf{p}_\ell(t)$ to the rational PH curve problem, we solve \eqref{eq:17} for $\lambda_{-5}$, $\lambda_{-4}$, and $\lambda_{-3}$ and obtain
  \begin{equation}
    \label{eq:19}
    \lambda_{-5} = \lambda_{-2},\
    \lambda_{-4} = -\lambda_{-1},\
    \lambda_{-3} = \lambda_{-2}.
  \end{equation}
  For $\ell \in \mathbb{Z} \setminus \{-5, -4, -3\}$ denote by $\mathbf{p}_{\ell}(t)$ the rational curve obtained as
  \begin{equation*}
    \mathbf{p}_\ell(t) = \int \lambda(t) \mathcal{A}(t)\qi\Cj{\mathcal{A}}(t) \dif t
  \end{equation*}
  where $\lambda_k = \delta_k^\ell$ (Kronecker delta) has been substituted into \eqref{eq:14} for $k \notin \{-5, -4, -3\}$ and with \eqref{eq:19} taken into account. A basis for the space of all rational PH curves with quaternionic pre-image $\mathcal{A}(t)$ and the single root $t = 0$ of the denominator polynomial is formed by the curves $\mathbf{p}_\ell$ for $\ell \in \mathbb{Z} \setminus \{-5, -4, -3\}$ plus the three constant ``curves'' $\qi$, $\qj$, $\qk$ whose linear combinations account for the integration constant or translation of solution curves.

  In order to obtain solutions $\mathbf{r}_\ell(t)$ to the rational arc length curve problem, we need to solve simultaneously \eqref{eq:17} and \eqref{eq:18} for $\lambda_{-5}$, $\lambda_{-4}$, $\lambda_{-3}$, and $\lambda_{-2}$. We obtain
  \begin{equation}
    \label{eq:20}
    \lambda_{-5} = -\lambda_{-1},\
    \lambda_{-4} = -\lambda_{-1},\
    \lambda_{-3} = -\lambda_{-1},\
    \lambda_{-2} = -\lambda_{-1}.
  \end{equation}
  For $\ell \in \mathbb{Z} \setminus \{-5, -4, -3, -2\}$ set
  \begin{equation*}
    \mathbf{r}_\ell(t) =
    \int \lambda(t) \mathcal{A}(t)\qi\Cj{\mathcal{A}}(t) \dif t
  \end{equation*}
  where $\lambda_k = \delta_k^\ell$ has been substituted into \eqref{eq:14} for $k \notin \{-5, -4, -3, -2\}$ and with \eqref{eq:20} taken into account. The curves $\mathcal{r}_\ell$ for $\ell \in \mathbb{Z} \setminus \{-5,-4,-3,-2\}$ together with the constant solutions $\qi$, $\qj$, and $\qk$ form a basis for the space of rational arc length curves to $\mathcal{A}(t)$ and the single root $t = 0$ of the denominator polynomial.

  The individual basis vectors/curves can be classified into several groups. For $\ell < -5$, we have $\mathbf{p}(t) = \mathbf{r}(t)$, for example
  \begin{equation}
    \label{eq:21}
    \begin{aligned}
    \mathbf{p}_{-7}(t) = \mathbf{r}_{-7}(t) &=
    -\frac{1}{3t^6}\qj - \frac{2}{5t^5}\qk + \frac{1}{4t^4}(\qi + 2\qk) - \frac{2}{3t^3}\qj - \frac{1}{2t^2}\qi,\\
    \mathbf{p}_{-6}(t) = \mathbf{r}_{-6}(t) &=
    -\frac{2}{5t^5}\qj - \frac{1}{2t^4}\qk + \frac{1}{3t^3}(\qi + 2\qk) - \frac{1}{t^2}\qj - \frac{1}{t}\qi.
    \end{aligned}
  \end{equation}
  For $\ell \ge 0$, we have $\mathbf{p}(t) = \mathbf{r}(t)$ and these curves are the usual \emph{polynomial} solutions, for example
  \begin{equation}
    \label{eq:22}
    \begin{aligned}
      \mathbf{p}_0(t) = \mathbf{r}_0(t) &=
      2t\qj + \qk t^2 - \frac{1}{3}(\qi + 2\qk) t^3 + \frac{1}{2}\qj t^4 + \frac{1}{5}\qi t^5,\\
      \mathbf{p}_1(t) = \mathbf{r}_1(t) &=
      \qj t^2 + \frac{2}{3}\qk t^3 - \frac{1}{4}(\qi + 2\qk)t^4 + \frac{2}{5}\qj t^5 + \frac{1}{6}\qi t^6.
    \end{aligned}
  \end{equation}
  For the curves in both, \eqref{eq:21} and \eqref{eq:22}, we see a typical band structure of non-zero coefficients. The band width is $2\deg\mathcal{A}(t) + 1$.

  The difference between rational PH curves and rational arc length curves lies in the basis vectors
  \begin{equation*}
    \begin{aligned}
    \mathbf{p}_{-2}(t) &=
    -\frac{1}{2t^4}\qj
    -\frac{2}{3t^3}\qk
    +\frac{1}{2t^2}(\qi - 2\qj + 2\qk)
    -\frac{2}{t}(2\qj + \qk)\\
      &\qquad\qquad\qquad
    -(\qi - 2\qj + 2\qk)t
    +\frac{1}{2}(\qi+2\qj)t^2
    +\frac{1}{3}\qi t^3,\\
    \mathbf{p}_{-1}(t) &=
    \frac{2}{3t^3}\qj
    +\frac{1}{t^2}\qk
    -\frac{1}{t}(\qi + 2\qk)
    -(\qi - 2\qk)t
    -\frac{1}{2}(\qi + 2\qk)t^2
    +\frac{2}{3}\qj t^3
    +\frac{1}{4}\qi t^4,\\
    \mathbf{r}_{-1}(t) &=
    \frac{1}{2t^4}\qj
    +\frac{2}{3t^3}(\qj+\qk)
    -\frac{1}{2t^2}(\qi - 2\qj)
    -\frac{1}{t}(\qi - 4\qj)\\
      &\qquad\qquad\qquad
    -2(\qj-2\qk)t
    -(\qi+\qj+\qk)t^2
    -\frac{1}{3}(\qi - 2\qj)t^3
    +\frac{1}{4}\qi t^4.
    \end{aligned}
  \end{equation*}

  The special solution $\mathbf{r}(t)$ of \eqref{eq:10} is a scaled and translated copy of $\mathbf{r}_{-1}(t)$:
  \begin{equation*}
    \mathbf{r}(t) = 300\mathbf{r}_{-1}(t) - 96\qi - 160\qj + 480\qk.
  \end{equation*}
\end{example}

\subsection{Zero Residue Approach of (Farouki and Sakkalis 2019)}
\label{FarSak}

Farouki and Sakkalis in \cite{FaroukiSakkalis2019} give a construction of a class of planar rational curves with rational arc-length function. They also suggest an extension to the spatial PH curves by constructing one example. Similarly to our method they impose a zero residue condition to a suitably chosen derivative vector. In this subsection we make a brief comparison and explain how our method compares to theirs and how it can be used to compute planar curves with rational arc length.

In \cite{FaroukiSakkalis2019} planar rational curves are constructed from three polynomials $u$, $v$, $w$ by integrating the expression
\begin{equation}
  \label{eq:23}
  \Bigl ( \frac{u^2-v^2}{w^2}, \frac{2uv}{w^2} \Bigr ).
\end{equation}
The authors show that if the integral is rational the resulting curves have automatically a rational arc-length function. The residue of \eqref{eq:23} is studied under the assumption of $w$ having only single roots and a vanishing condition is given using certain polynomial divisibility criteria involving $u$, $v$ and $w$. It is shown that this condition can be non-trivially satisfied for sufficiently high polynomial degrees.

\begin{example}
  \label{ex:FaroukiSakkalis2019-b1}
  In \cite[Example~1]{FaroukiSakkalis2019} the following planar curve with rational arc length is constructed
  \begin{equation}
    \label{eq:24}
    \mathbf r(t) = \frac{1}{3t(t-1)}(
        t^5\qi -(4\qi - 3\qj)t^4
      - (3\qi + 9\qj)t^3
      + 7t^2\qi
      + 2t\qi
      - 3\qi + 12\qj
    ).
  \end{equation}
  Let us reconstruct this example using our residue method. It can be done using the polynomials $\mathcal A(t) = t^3 + (\qk - 2)t^2 - \qk t + \qk$ and $\alpha(t) = t(t-1)$. This choice is noteworthy in two regards:
  \begin{itemize}
  \item $\mathcal A(t)$ is in the sub-algebra spanned by $1$ and $\qk$ which is isomorphic to $\mathbb{C}$. Consequently $\mathbf{F}(t) \coloneqq \mathcal{A}(t)\qi\Cj{\mathcal{A}}(t)$ lies in the space spanned by $\qi$, $\qj$ and the resulting PH curves are planar. In fact all planar PH curves can be produced from this kind of restricted quaternion pre-images.
  \item Since $\mathcal{A}(t)$ is $\qi$-reduced, Theorem~\ref{th:linear-system} tells us that non-polynomial solutions are only possible if there are certain linear dependencies between the first coefficients of the Taylor series of $\mathbf{F}(t) \coloneqq \mathcal{A}(t)\qi\Cj{\mathcal{A}}(t)$ at the zeros $t = 0$ or $t = 1$ of $\alpha(t)$. Indeed, we even have
    \begin{equation}
      \label{eq:25}
      \mathcal F(t) = (1-\qi) - 2(1-\qi) t + \mathcal O(t^2),
      \quad
      \mathcal F(t) = 2(1-\qj) + 4(1-\qj)(t - 1) + \mathcal O((t-1)^2).
    \end{equation}
    where $\mathcal{F}(t) \coloneqq \mathcal{A}(t)(1+\qi)\mathcal{A}(t)$. This shows that the first two Taylor coefficients of $\mathcal F(t)$ (and hence also of $\mathbf{F}(t)$) at both zeros are linearly dependent.
  \end{itemize}
  Equation~\eqref{eq:25} suggests the partial fraction decomposition
  \begin{equation*}
    \lambda(t) = \sum_{\ell=-k_0}^{-1} \lambda_{0,\ell}t^\ell + \sum_{\ell=-k_1}^{-1} \lambda_{1,\ell}(t-1)^\ell
  \end{equation*}
  with $k_0$, $k_1 \ge 2$ and yet to be determined coefficients $\lambda_{0,\ell}$, $\lambda_{1,\ell}$ for the rational function $\lambda(t)$ of \eqref{eq:4} and \eqref{eq:5}. With this choice, the zero residue conditions \eqref{eq:13} become
  \begin{equation*}
    (\lambda_{0,-1} - 2\lambda_{0,-2}) (1-\qi) = 0,
    \quad
    2 (\lambda_{1,-1} + 2 \lambda_{1,-2}) (1-\qj) = 0
  \end{equation*}
  and give
  \begin{equation*}
    \lambda(t) = \lambda_{0,-2}\Bigl(\frac{2}{t}+\frac{1}{t^2}\Bigr) + \lambda_{1,-2}\Bigl(\frac{2}{t-1} - \frac{1}{(t-1)^2}\Bigr)
  \end{equation*}
  for arbitrary $\lambda_{0,-2}$, $\lambda_{1,-2} \in \mathbb{R}$. For the sake of a low degree denominator, we choose $k_0 = k_1 = 2$. The general solution is given by the vector part of the integral
  \begin{equation*}
    \mathcal{R}(t) = \int (\lambda_{0,-2}(2t^{-1}+t^{-2}) + \lambda_{1,-2}(2(t-1)^{-1}-(t-1)^{-2})) \mathcal{F}(t) \dif t.
  \end{equation*}
  Two particular solutions are given by $\lambda_{0,-2} = 1$, $\lambda_{1,-2} = 0$ and $\lambda_{0,-2} = 0$, $\lambda_{1,-2} = 1$, respectively:
  \begin{multline*}
    \mathbf{r}_0(t) = \frac{1}{t}\bigl(\qi + (\qi-4\qj)t^2 + (-2\qi-\qj)t^3
    + \tfrac{1}{3}(7\qi+6\qj)t^4 + \tfrac{1}{2}(\qi-5\qj)t^5 +
    \tfrac{1}{5}(-7\qi+4\qj)t^6 + \tfrac{1}{3}\qi t^7
    \bigr),
  \end{multline*}
  \begin{multline*}
    \mathbf{r}_1(t) = 2\qj (t-1)^{-1} - \tfrac{1}{30}(103\qi-111\qj) - (4\qi-6\qj)(t-1)^1 + (\qi+3\qj)(t-1)^2 +\\ \tfrac{10}{3}\qi(t-1)^3 + \tfrac{3}{2}(\qi-\qj)(t-1)^4 + \tfrac{1}{5}(-3\qi-4\qj)(t-1)^5 - \tfrac{1}{3}\qi(t-1)^6.
  \end{multline*}
  The solution \eqref{eq:24} of \cite[Example~1]{FaroukiSakkalis2019} is obtained as $\mathbf{r}(t) = \mathbf{r}_0(t) + \mathbf{r}_1(t) + \frac{1}{3}\qi - 2\qj$.
\end{example}

We see the fundamental difference between our method and \cite{FaroukiSakkalis2019}. While we assume only $\mathcal A(t)$ to be given and describe all possible $\lambda(t)$, Farouki and Sakkalis study the residue condition relating all input polynomials $u$, $v$, $w$. This approach leads to a restriction on $w$ and the resulting curves can have only single roots in the denominator as in \eqref{eq:24}. Our method provides the complete space of solutions and the restriction of \cite{FaroukiSakkalis2019} can be interpreted via a special choice for $\lambda(t)$ by imposing $k_i=2$ in~\eqref{eq:11}.

\begin{example}
  \label{ex:FaroukiSakkalis2019-b2}
  In \cite[Example~9]{FaroukiSakkalis2019} we also find one example of a rational space curve with a rational arc length. We are going to recover it using the zero residue conditions. The denominator polynomial is $\alpha = t(t-1)$ and we should use
  \begin{equation*}
    \mathcal A(t) = (\qi - 2\qj - \qk)t^3 + (-1 -2\qi + 3\qj + 4\qk)t^2
    + (1 - 2\qi + \qj -2\qk)t  - 1 + 2\qi - \qj + 2\qk.
  \end{equation*}
  The remaining computation is quite similar to our Example~\ref{ex:FaroukiSakkalis2019-b1}. Because of
  \begin{equation*}
    \begin{aligned}
      \mathcal{F}(t) &= \mathcal{A}(t)(1+\qi)\Cj{\mathcal{A}}(t) \\
                     &= (10 - 8\qj + 6\qk) + (-20 + 16\qj - 12\qk)t + \mathcal{O}(t^2) \\
                     &= (12-8\qi-4\qj+8\qk) + (24-16\qi-8\qj+16\qk)t + \mathcal{O}((t-1)^2),
    \end{aligned}
  \end{equation*}
  the zero residue constraints
  \begin{equation*}
    \begin{aligned}
      \lambda_{0,-1}(10-8\qj+6\qk) + \lambda_{0,-2}(-20+16\qj-12\qk) &= 0,\\
      \lambda_{1,-1}(12-8\qi-4\qj+8\qk) + \lambda_{1,-2}(24-16\qi-8\qj+16\qk) &= 0
    \end{aligned}
  \end{equation*}
  can be satisfied by the ansatz
  \begin{equation*}
    \lambda = \frac{\lambda_{0,-1}}{t} + \frac{\lambda_{0,-2}}{t^2} + \frac{\lambda_{1,-1}}{t-1} + \frac{\lambda_{1,-2}}{(t-1)^2}.
  \end{equation*}
  and lead to $\lambda_{0,-1} = 2\lambda_{0,-2}$, $\lambda_{1,-1} =
  -2\lambda_{1,-2}$. The particular solutions to $\lambda_{0,-2} = 1$,
  $\lambda_{1,-2} = 0$ and $\lambda_{0,-2} = 0$, $\lambda_{1,-2} = 1$ are
  \begin{multline*}
    \mathbf r_0(t) = \frac{1}{15t}
    \bigl(
    120\qj-90\qk
    -(120\qi-360\qj-30\qk)t^2
    -(30\qi+180\qj-210\qk)t^3 \\
    -(20\qi+80\qj+190\qk)t^4
    -(135\qi-30\qj-90\qk)t^5
    +(108\qi+36\qj+18\qk)t^6
    +(-20\qi-20\qj-10\qk)t^7
    \bigr),
  \end{multline*}
  \begin{multline*}
    \mathbf r_1(t) = \frac{1}{15(t-1)}\bigl(
      120 \qi + 60 \qj - 120 \qk + (287 \qi - 106 \qj + 22 \qk)(t+1) +
      (480 \qi + 60 \qj - 120 \qk) (t - 1)^2\\
      + (150 \qi + 120 \qj - 240 \qk) (t - 1)^3  - (140 \qi + 20 \qj + 160 \qk)
      (t - 1)^4\\  - (105 \qi - 90 \qj + 30 \qk) (t - 1)^5
      + (12 \qi + 84 \qj + 42 \qk) (t - 1)^6  + (20 \qi + 20 \qj + 10 \qk) (t - 1)^7
    \bigr)
  \end{multline*}
  The solution of \cite[Example~9]{FaroukiSakkalis2019} is $\mathbf{r}_0(t) + \mathbf{r}_1(t) + \frac{1}{3}(-10\qi + 38\qj + 4\qk)$.
\end{example}

\section{Third Way -- A Geometric Dual Approach}
\label{sec:geometric-approach}

Computing rational arc length curves by solving the system of linear equations resulting from \eqref{eq:6} or by satisfying the zero residue constraints \eqref{eq:13} are algebraic in nature. In this section we present a predominantly geometric approach that generalizes the construction for planar rational arc length curves of \cite{Pottmann95} to space curves.

In \cite[Theorem~2.3]{Pottmann95}, Pottmann proved that the planar rational curves with rational arc length are precisely the evolutes of planar rational PH curves. Unfortunately, this nice result does not generalize to space curves, at least not in an obvious way. However, Pottmann also describes a construction for planar rational arc length curves in terms of spatial curves of constant slope (defined as curves having constant angle between the tangent vector and a fixed direction). This construction allows for the following spatial generalization:

\begin{theorem}
  \label{th:slope-one}
  The rational curve $\mathbf r(t) \subset \IM \mathbb{H}$ has a piecewise rational arc length function $s(t)$ if and only if it is PH and the curve $\mathcal R(t) \coloneqq s(t) + \mathbf r(t) \subset \mathbb{H}$ is rational and of constant slope $ 1$ with respect to the real coordinate direction.
\end{theorem}

\begin{proof}
  A rational arc length function implies the PH property by Lemma~\ref{lem:PH}. Moreover, the tangent vector $\dot{\mathcal R}(t) = \dot s(t) + \dot{\mathbf{r}}(t) = \sqrt{\dot{\mathbf{r}}(t) \cdot \dot{\mathbf{r}}(t)} + \dot{\mathbf{r}}(t)$ has, indeed, slope $ 1$ with respect to the real coordinate direction.

  If, conversely, $\mathbf{r}(t)$ is PH and $\mathcal{R}(t)$ is of constant slope $ 1$, then the length of $\dot{\mathbf r}(t)$ equals the length of $\dot{s}(t)$. But then $s(t) = \int \dot{s}(t) \dif t = \int \sqrt{\dot{\mathbf r}(t) \cdot \dot{\mathbf r}(t)} \dif t$ is the arc length of $\mathbf{r}(t)$ and it is rational.
\end{proof}

\begin{remark}
  The rational curve $\mathcal R(t) \subset \mathbb{H}[t]$ of constant slope $ 1$ in Theorem~\ref{th:slope-one} is a PH curve as well because $\dot{\mathcal R}(t) \cdot \dot{\mathcal R}(t) = \dot{s}^2 + \dot{\mathbf r}(t) \cdot \dot{\mathbf r}(t) = 2\dot{s}^2$ is a square in~$\mathbb{R}[t]$.
\end{remark}

The basic idea of this section is to compute the rational PH curve $\mathcal R(t) = s(t) + \mathbf r(t)$ of constant slope $ 1$ in $\mathbb{H}$ as envelope of osculating hyperplanes. In this sense, our approach can not only be viewed as a generalization of \cite{pottmann95b} from planar to space curves but also as an extension of the dual approach of \cite{FaroukiSir} from rational PH curves to rational PH curves of constant slope in dimension four. By $\langle \mathbf{x}, \mathbf{y} \rangle = \frac{1}{2}(\mathbf{x}\Cj{\mathbf{y}}+\Cj{\mathbf{y}}\mathbf{x})$ we denote the usual Euclidean scalar product in $\mathbb{R}^4 \cong \mathbb{H}$ or in $\mathbb{R}^3 \cong \IM\mathbb{H}$.

\begin{lemma}
  \label{lem:envelope}
  For a polynomial $\mathcal A(t) \in \mathbb{H}[t]$ define $\mathcal F(t) = \mathcal A(t)(1 + \qi)\Cj{\mathcal A}(t)$. Pick an arbitrary rational function $h(t) \in \mathbb{R}(t)$ and a quaternion polynomial $\mathcal N(t) \in \mathbb{H}[t]$ such that $\langle \mathcal N(t), \mathcal F(t) \rangle = \langle \mathcal N(t), \dot{\mathcal F}(t) \rangle = \langle \mathcal N(t), \ddot{\mathcal F}(t) \rangle = 0$ and consider the one-parameter family of hyperplanes with equation $H(t)\colon \langle \mathcal N(t), \mathcal X \rangle + h(t) = 0$ (where $\mathcal X = x_0 + x_1\qi + x_2\qj + x_3\qk$ denotes an undetermined coordinate vector in $\mathbb H$). Then, $\mathcal R(t) \coloneqq H(t) \cap \dot H(t) \cap \ddot H(t) \cap \dddot H(t)$ is a rational curve of constant slope $ 1$. Conversely, any rational curve $\mathcal R(t)$ of constant slope $ 1$ can be obtained in that way.
\end{lemma}

\begin{corollary}
  \label{cor:envelope}
  The vector part $\mathbf r(t) = \IM\mathcal R(t)$ of any rational curve $\mathcal R(t)$ constructed as in Lemma~\ref{lem:envelope} has a piecewise rational arc length function. Conversely, any rational curve in $\IM \mathbb{H}$ with this property can be obtained in that way.
\end{corollary}

\begin{proof}[Proof of Lemma~\ref{lem:envelope}]
  Since $\mathcal{R}(t)$ is obtained by solving a system of linear equations whose coefficients are polynomials, it is rational by construction. Its tangent is the intersection of the hyperplanes $H(t)$, $\dot H(t)$, and $\ddot H(t)$. Taking derivatives on both sides of $\langle \mathcal N(t), \mathcal F(t) \rangle = 0$ and simplifying using the defining conditions for $\mathcal N(t)$ we find
  \begin{equation*}
    \langle \mathcal N(t), \mathcal F(t) \rangle =
    \langle \dot{\mathcal N}(t), \mathcal F(t) \rangle =
    \langle \ddot{\mathcal N}(t), \mathcal F(t) \rangle = 0.
  \end{equation*}
  The solution space of the system of inhomogeneous linear equations (the tangent)
  \begin{equation*}
    \langle \mathcal N(t), \mathcal X \rangle + h =
    \langle \dot{\mathcal N}(t), \mathcal X \rangle + \dot h =
    \langle \ddot{\mathcal N}(t), \mathcal X \rangle + \ddot h = 0
  \end{equation*}
  is parallel to $\mathcal F(t)$ and because $N(\IM\mathcal F(t)) = (\RE\mathcal F(t))^2$ by definition of $\mathcal F(t)$, of constant slope~$1$.

  Conversely, assume that $\mathcal R(t)$ is rational of constant slope $ 1$ with respect to the real coordinate direction. It is obtained as intersection $H(t) \cap \dot H(t) \cap \ddot H(t) \cap \dddot H(t)$, where $H(t)$ is the osculating hyperplane of $\mathcal R(t)$. The hyperplane $H(t)$ contains $\mathcal R(t)$ and is parallel to $\dot{\mathcal R}(t)$, $\ddot{\mathcal R}(t)$, and $\dddot{\mathcal R}(t)$. It has an algebraic equation of type $\langle \mathcal N(t), \mathcal X \rangle + h(t) = 0$ with some polynomial $\mathcal N(t) \in \mathbb{H}[t]$ and some rational function $h(t) \in \mathbb{R}(t)$. By construction, $\langle \mathcal N(t), \dot{\mathcal R}(t) \rangle = \langle \mathcal N(t), \ddot{\mathcal R}(t) \rangle = \langle \mathcal N(t), \dddot{\mathcal R}(t) \rangle = 0$. Moreover, the curve $\mathcal R(t)$ has a polynomial field $\mathcal F(t) \in \mathbb{H}[t]$ of tangent vectors that, because of the constant slope assumption, can be written as $\mathcal F(t) = \mathcal A(t)(1+\qi)\Cj{\mathcal A}(t)$ for some $\mathcal A(t) \in \mathbb{H}(t)$. Because $\mathcal F(t)$ equals $\lambda(t) \dot{\mathcal R}(t)$ for some rational function $\lambda(t)$ we have
  \begin{equation*}
    \begin{aligned}
    \langle \mathcal N, \mathcal F(t) \rangle &= \langle \mathcal N, \lambda(t) \dot{\mathcal R}(t) \rangle = 0,\\
      \langle \mathcal N, \dot{\mathcal F}(t) \rangle &= \langle \mathcal N, \dot\lambda(t) \dot{\mathcal R}(t) + \lambda(t)\ddot{\mathcal R}(t) \rangle = 0,\\
    \langle \mathcal N, \ddot{\mathcal F}(t) \rangle &= \langle \mathcal N, \ddot\lambda(t) \dot{\mathcal R}(t) + 2\dot\lambda(t)\ddot{\mathcal R}(t) + \lambda(t)\dddot{\mathcal R}(t) \rangle = 0.
    \end{aligned}
  \end{equation*}
  This concludes the proof.
\end{proof}

\begin{example}
  \label{ex:envelope}
  We pick the polynomial $\mathcal{A}(t) = t^2 + \qk t + \qi + \qj$ of previous examples and the rational function $h(t) = t^{-1}$.

  A non-trivial solution for $\mathcal N(t)$ is found for $\deg \mathcal N(t) = 6$:
  \begin{multline*}
    \mathcal N(t) =
    (1 -\qi + \qk) t^6
    - 3(1 - \qi) t^5
    + 3(\qj + 2 \qk) t^4
    - (9 - 7 \qi - 8 \qj + 8 \qk) t^3 \\
    + 3(2 - 4 \qi - 2 \qj + \qk) t^2
    + 6 \qi t
    - 1 -\qi + \qj.
  \end{multline*}
  It is unique up to a constant real factor and results in the rational arc length curve
  \begin{multline}
    \label{eq:26}
    \mathbf r(t) =
    \frac{1}{6t^4(t^4+4t^3-6t^2+4t-1)^3}
    \bigl(
      105 \qi t^{12}
      + 168(4 \qi + \qj) t^{11}
      + 14(21 \qi + 74 \qj - 10 \qk) t^{10} \\
      - 60(45 \qi - 10 \qj + 12 \qk) t^9
      + 6(439 \qi - 529 \qj + 30 \qk) t^8
      - 4(181 \qi - 1096 \qj - 760 \qk) t^7\\
      - 12(42 \qi + 161 \qj + 405 \qk) t^6
      + 12(47 \qi - 74 \qj + 310 \qk) t^5
      - 5(51 \qi - 304 \qj + 348 \qk) t^4\\
      + 12(5 \qi - 74 \qj + 44 \qk) t^3
      - 6(\qi - 52 \qj + 16 \qk) t^2
      - 8(8 \qj - \qk) t
      + 6 \qj
    \bigr).
  \end{multline}
  Note the denominator factors of relatively high multiplicity which are in accordance with previous examples. In the light of Examples~\ref{ex:linear-system} and \ref{ex:zero-residue} it is also natural to ask for a choice of $h(t)$ that reproduces \eqref{eq:10} or the curves of \eqref{eq:21}. It is easy to recover the rational function $h(t)$ from a solution curve, for example from the rational arc length curve $\mathbf{r}_{-7}(t)$ given in \eqref{eq:21}:
  \begin{equation*}
    h(t) =
    \frac{-t^3(3t^4+10t^3-12t^2+6t-1)}{30(t^6-3t^5-9t^3+6t^2-1)}
  \end{equation*}
  It is, however, unclear how to select the $h(t)$ a priori so that a ``simple'' solution of type \eqref{eq:21} is obtained.
\end{example}

\section{Comparison of Different Methods}
\label{sec:comparison}

We have presented three methods to compute rational curves with rational arc length and a prescribed polynomial field of tangent vectors, given by a quaternion polynomial $\mathbf{F}(t) = \mathcal{A}(t)\qi\Cj{\mathcal{A}}(t) \subset \IM\mathbb{H}[t]$:
\begin{itemize}
\item In our first approach (Theorem~\ref{th:linear-system}), we prescribe the denominator polynomial $\alpha(t)$. The numerator polynomial is found by solving a system of linear equations.
\item In our second approach (Theorem~\ref{th:zero-residue}), we prescribe the denominator polynomial $\alpha(t)$. The numerator polynomial is found by imposing zero residue constraints on the partial fraction coefficients of the rational function $\lambda(t)$ of Equations~\eqref{eq:4} and~\eqref{eq:5}.
\item In our third approach (Lemma~\ref{lem:envelope}, Corollary~\ref{cor:envelope}), we prescribe a rational function $h(t)$ and compute the curve as the projection of an envelope of hyperplanes in the four-dimensional vector space~$\mathbb{H}$.
\end{itemize}
We will refer to these methods either as ``linear system approach'', ``zero residue approach'' or ``dual approach'' or simply as ``first'', ``second'' and ``third approach''. In what follows, we are going to discuss and compare these three methods with each other and also with prior attempts to construct rational PH curves with rational arc length. Since all three methods can be viewed as generalizations of computation algorithms for rational PH curves \cite{kalkan22,schroecker23,SchroeckerSir:_optimal_interpolation,Pottmann95,FaroukiSir}, the scope of this discussion extends to rational PH curves without the rational arc length constraint.

Let's first talk about specifics of the rational arc length case. The most important advantage of our approaches in that regard is \emph{universality.} All three methods are capable of producing any rational curve with rational arc length. This is not the case for the curves presented in \cite{LEE2014689,FaroukiSakkalis2019} and discussed in \cite{FAROUKI20151}. Not only do we primarily consider spatial curves (with the possibility to specialize for the planar case, cf. Example~\ref{ex:FaroukiSakkalis2019-b1}), we also cover cases with non-simple roots in the denominator. In fact, we demonstrate that for algebraically independent quaternionic pre-image $\mathcal A(t)$ and denominator polynomial $\alpha(t)$, at least one zero of $\alpha(t)$ needs to have multiplicity $n = 4$ in order to allow for non-polynomial PH curves with rational arc length. The multiplicity $n$ is actually related to the dimension $d$ of the ambient space via $n = d + 1$ (cf. Remark~\ref{rem:dimension}). If the rational arc length requirement is dropped, $n = d$ suffices to ensure non-polynomial PH curves.

A remarkable insight is that rational PH curves with or without rational arc length either have algebraically dependent (non-generic) quaternionic pre-image $\mathcal A(t)$ and denominator polynomial $\alpha(t)$ or denominator polynomials $\alpha(t)$ with multiple zeros (non-generic again). Much research in the past \cite{krajnc1,krajnc2,krajnc3,FaroukiSakkalis2019} was devoted to the first of the two cases with the reputable main aim to achieve low curve degrees but usually without mentioning the algebraic dependence. Our first and second approach elucidate this in terms of linear dependencies between low degree Taylor coefficients of $\mathcal{F}(t)$ at zeros of $\alpha(t)$. We nourish some hope that this formulation will allow to simplify otherwise rather complicated constructions for suitable quaternionic pre-image polynomials~$\mathcal{A}(t)$.

A notable difference between our methods to a large body of literature on PH curves (including polynomial curves and curves with non-rational arc length) is that we consider the pre-image polynomial $\mathcal{A}(t)$ to be part of the input while often it is to be computed from certain data. This suggests combined approaches where $\mathcal{A}(t)$ is computed in some way that ensures desirable properties (Hermite interpolation conditions, optimal convergence etc.) and then used as input to our methods. Let us illustrate this at hand of an example:

\begin{example}
  We would like to solve a $C^1$ Hermite interpolation problem. In order to re-use results of Examples~\ref{ex:linear-system} and \ref{ex:zero-residue}, we sample the tangent directions for $t_0 = \frac{1}{2}$, $t_1 = 2$ from $\mathbf{F}(t) = \mathcal{A}(t)\qi\Cj{\mathcal{A}}(t)$ where $\mathcal{A}(t)$ is taken from \eqref{eq:9}. Note that this merely avoids duplicating the computations of Example~\ref{ex:zero-residue}. The presented method is completely general.

  Let us consider the data consisting of the end points
  \begin{equation*}
    \mathbf{p}_0 = 0,\quad \mathbf{p}_1 = \tfrac{1}{10}(3\qi - 70\qj - 9\qk),
  \end{equation*}
  corresponding unit tangent directions
  \begin{equation*}
    \mathbf{d}_0 = \tfrac{1}{37}(3\qi - 36\qj - 8\qk),\quad \mathbf{d}_1 = \tfrac{1}{11}(-6\qi - 9\qj + 2\qk)
  \end{equation*}
  and parameter values $t_0 = \frac{1}{2}$, $t_2 = 2$. It is to be interpolated by a rational arc length curve $\mathbf{r}(t)$. The conditions to be fulfilled are $\mathbf{r}(t_\ell) = \mathbf{p}_\ell$, $\dot{\mathbf{r}}(t_\ell) = \mathbf{d}_\ell$ for $\ell \in \{0,1\}$.

  The first step consists of finding a suitable quaternionic polynomial $\mathcal{A}(t)$ such that $\mathbf{F}(t_\ell)$ is a multiple of $\mathbf{d}_\ell$ for $\ell \in \{0,1\}$. This problem is standard for $G^1$ or $C^1$ interpolation and will not be discussed here, see e.g. \cite{FaroukiSir,krajnc2}. By construction of the data, we can base our solution to the Hermite interpolation problem on the polynomial $\mathcal{A}(t)$ of \eqref{eq:9}.

  In a next step, we define a vector space of rational arc-length curves to $\mathcal{A}(t)$ to pick our solution from. Since $t = 0$ is outside of our interpolation interval $[t_0, t_1]$, we may use the linear span of the rational arc-length curves $\mathbf{r}_{\ell}(t)$, $\ell \in \mathbb{Z} \setminus \{-5, -4, -3, -2\}$ and $\qi$, $\qj$, $\qk$ of Example~\ref{ex:zero-residue}. Imposing the Hermite interpolation conditions results in eight linear constraints, two times three conditions for point interpolation and two times one condition in order to match length of derivative vectors. Thus, a solution of the shape
  \begin{equation*}
    \mathbf{r}(t) =
    \varrho_{-7}\mathbf{r}_{-7}(t) +
    \varrho_{-6}\mathbf{r}_{-6}(t) +
    \varrho_{-1}\mathbf{r}_{-1}(t) +
    \varrho_{0}\mathbf{r}_{0}(t) +
    \varrho_{1}\mathbf{r}_{1}(t) +
    \theta_1 \qi +
    \theta_2 \qj +
    \theta_3 \qk
  \end{equation*}
  can be expected. Indeed, we find
  \begin{gather*}
    \varrho_{-7} = \tfrac{1269492328}{584801613},\quad
    \varrho_{-6} = -\tfrac{2982893698}{584801613},\quad
    \varrho_{-1} = -\tfrac{488565929}{584801613},\quad
    \varrho_{0} = -\tfrac{82611679}{194933871},\quad
    \varrho_{1} = \tfrac{45535106}{194933871},\\
    \theta_1 = -\tfrac{39697738537}{8772024195},\quad
    \theta_2 = -\tfrac{15049793575}{2079294624},\quad
    \theta_3 = \tfrac{130723703147}{46784129040}.
  \end{gather*}
  The resulting curve is depicted in Figure~\ref{fig:hermite}. This example demonstrates that, once $\mathcal{A}(t)$ is suitably chosen, the interpolation problem becomes linear. Thus, it is possible to combine it with optimal interpolation algorithms and shape control (no cusps) of \cite{SchroeckerSir:_optimal_interpolation}. Basically, the rational arc length constraint just adds linear equations to a linear program. Extensions to higher order Hermite interpolation problems and relaxation to $G^1$ interpolation are possible as well.
\end{example}

\begin{figure}
  \centering
  \includegraphics[]{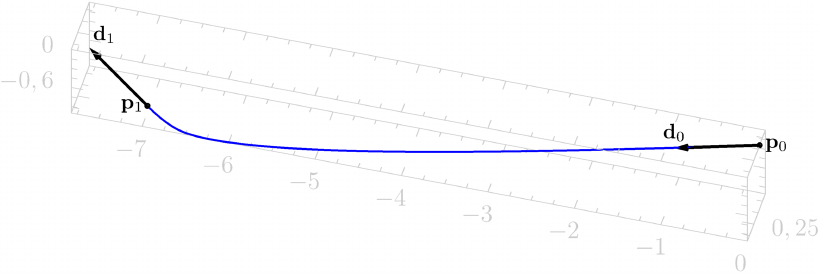}
  \caption{Hermite interpolation by a curve with rational arc length.}
  \label{fig:hermite}
\end{figure}

When comparing our three approaches with each other it is evident that the dual approach allows for little control over the degree of the resulting curve $\mathbf r(t)$ (cf. Example~\ref{ex:envelope}). Both other methods perform better in that regard but generally require that the necessary conditions for existence of rational (non-polynomial) solutions are satisfied. It is not obvious how to obtain polynomial PH curves with the envelope approach while the linear system method yields only polynomial solutions unless the criterion of Theorem~\ref{th:linear-system} is met. With the zero residue approach, both polynomial and non-polynomial curves can be generated systematically by satisfying the zero residue constraints \eqref{eq:13} with $\lambda_{i,k} = 0$ for suitable indices $i$ and~$k$.

Our first or second approach highlight the existence of an underlying vector space of solutions (that is of high importance in \cite{SchroeckerSir:_optimal_interpolation}), as we need to solve a system of linear equations. This vector space is also inherent in the third approach because linear combinations of rational functions $h(t)$ give rise to linear combinations of the resulting solution curves. This vector space structure is typically hidden in more traditional approaches that focus on particular solutions. While \cite[Example~9]{FaroukiSakkalis2019} is an isolated example of a rational space curve with rational arc length, we naturally view it as element of a vector space of dimension two (if translations are factored out) in Example~\ref{ex:FaroukiSakkalis2019-b2}.

In the actual design with PH curves (with or without rational arc length) some sort of shape control is of paramount importance. In particular, it is necessary and notoriously difficult to avoid zeros of the derivative vector $\dot{\mathbf r}(t)$ that generically lead to cusps. These are precisely the zeros of the real polynomial $\mu(t)$ that appears in \eqref{eq:8}, thus showing that it is not just a by-product of the linear system approach but has important meaning. Note that it can also be easily extracted from the zero residue approach. With other approaches, this is not always straightforward. The task can be reduced to computing the $\gcd$ of some real polynomials which is tricky if polynomials are given in numeric form or by some involved non-rational symbolic expressions. This is surprisingly often the case when using literature recommendations for the pre-image polynomial~$\mathcal{A}(t)$.

Summarizing, we acknowledge the benefits of traditional procedures to produce rational PH curves and rational PH curves with rational arc lengths that are of a \emph{low degree.} We consider our methods superior when it comes to more general PH curves and we particularly favor the linear system or zero residue approach when it comes to shape control. In fact, we feel that the zero residue approach allows for the most systematic and clear construction of structured bases for solution vector spaces (cf. \cite[Figure~1]{SchroeckerSir:_optimal_interpolation}).

\section{Conclusion}
\label{sec:conclusion}

We solved by three different methods the open problem of constructing all spatial rational curves with rational arc-length function, discussed their merits and compared them with previous computation attempts. We also provided a new proof of the characterization in Lemma~\ref{lem:factorization}. In the future we plan to use the three methods for solving various problems related to PH curves, such as the construction of smooth closed PH curves with prescribed length. We feel that the rational spatial PH curves are very important in general, because of their ability to represent framing motions of curves with rational trajectories \cite{kalkan22} and further desirable properties.

\bibliographystyle{elsarticle-num}

\end{document}